\documentclass[submission,copyright,creativecommons]{eptcs}

\usepackage{amsthm, amsmath, amssymb, dsfont, physics, textcomp, tikz}
\usepackage{braket}
\usepackage{qcircuit}
\usepackage{tikz-cd}
\RequirePackage{doi}
\usepackage{hyperref}
\usepackage{stmaryrd}
\usepackage[capitalize,nameinlink,sort,noabbrev]{cleveref}
\usepackage{algorithm}
\usepackage{algpseudocode}
\usepackage{booktabs}
\usepackage{zx-calculus}
\usepackage{caption}
\usepackage{subcaption}
\usepackage{underscore}

\newcommand{\Z}{\mathbb{Z}}

\newcommand{\C}{\mathbb{C}}

\newcommand{\R}{\mathbb{R}}

\renewcommand{\H}{\mathcal{H}}

\newcommand{\x}{\vec{x}}
\newcommand{\y}{\vec{y}}

\newcommand{\hgate}{\textsc{h}}
\newcommand{\tgate}{\textsc{t}}
\newcommand{\xgate}{\textsc{x}}

\newtheoremstyle{break}
  {}
  {}
  {\itshape}
  {}
  {\bfseries}
  {.}
  {\newline}
  {}

\theoremstyle{plain}
\newtheorem{theorem}{Theorem}[section]

\newtheorem{proposition}[theorem]{Proposition}
\newtheorem*{proposition*}{Proposition}

\theoremstyle{break}

\theoremstyle{definition}
\newtheorem{definition}[theorem]{Definition}

\newtheorem{example}[theorem]{Example}
\theoremstyle{remark}

\newcommand{\urlalt}[2]{\href{#2}{\nolinkurl{#1}}}

\usepackage{xcolor}
\hypersetup{
    colorlinks,
    linkcolor={red!80!black},
    citecolor={green!80!black},
    urlcolor={blue!80!black}
}

\providecommand{\event}{QPL 2023}

\newcommand{\titlerunning}{Complete Equational Theories for the Sum-Over-Paths with Unbalanced Amplitudes}
\newcommand{\authorrunning}{M. Amy}

\title{Complete Equational Theories for the Sum-Over-Paths with Unbalanced Amplitudes}

\author{Matthew Amy
\institute{School of Computing Science \\ Simon Fraser University }
\email{meamy@sfu.ca}
}

\begin{document}

\maketitle

\begin{abstract}

Vilmart recently gave a complete equational theory for the balanced sum-over-paths over Toffoli-Hadamard circuits, and by extension Clifford+$\mathrm{diag}(1, \zeta_{2^k})$ circuits. Their theory is based on the phase-free ZH-calculus which crucially omits the average rule of the full ZH-calculus, dis-allowing the local summation of amplitudes. Here we study the question of completeness in unbalanced path sums which naturally support local summation. We give a concrete syntax for the unbalanced sum-over-paths and show that, together with symbolic multilinear algebra and the interference rule, various formulations of the average and ortho rules of the ZH-calculus are sufficient to give complete equational theories over arbitrary rings and fields. 

\end{abstract}

\section{Introduction}

The balanced sum-over-paths representation of a linear operator $\Psi:\C^{2^m}\rightarrow \C^{2^n}$ introduced in \cite{a18} is a symbolic representation of $\Psi$ over Boolean-valued variables, having the form
\begin{equation}
  \label{eq:pathsum1}
	\Psi\ket{\vec{x}} = \mathcal{N}\sum_{\vec{y}\in\Z_2^{k}}e^{2\pi i P(\vec{x}, \vec{y})}\ket{f(\vec{x}, \vec{y})}
\end{equation}
where $P:\Z_2^m\times\Z_2^k\rightarrow \R/2\pi$ and $f:\Z_2^m\times\Z_2^k\rightarrow \Z_2^n$ are represented by (systems of) polynomials in $m+k$ variables. The expression of \cref{eq:pathsum1} is interpreted as a sum over the \emph{paths} a system may take beginning from some initial configuration $\vec{x}$. If $\Psi$ is taken as encoding some evolution of a physical ($2^k$-level) system, the expression \cref{eq:pathsum1} coincides roughly with Richard Feynman's \emph{path integral} \cite{fh65}. As with Feynman's path integral the sum in \cref{eq:pathsum1} is \emph{balanced}, in that each path --- indexed by the values of $\vec{x}$ and $\vec{y}$ --- has the same amplitude $\mathcal{N}$ but varies in the phase $e^{2\pi i P(\vec{x}, \vec{y})}$. As standard operators used in quantum computation can be represented in this form and the representation is closed over composition and tensor products, typical quantum computational processes admit representations by balanced path sums.

It was previously shown \cite{a18} that with a concrete representation in terms of polynomials, the balanced sum-over-paths admits a simple equational theory. This equational theory was shown to be relatively complete, and with a small modification complete \cite{v21}, for Clifford operators. As the number of superfluous variables $y_i$ in the sum, called \emph{internal} and corresponding to a pair of paths between two end-points, offers an intuitive notion of complexity of the expression, the equational theory further gives rise to a natural re-write system which iteratively removes such variables from the sum. This re-write system was shown to terminate in polynomial time with a unique normal form for Clifford operators \cite{abr22}, and was further shown to perform well in practice on Clifford+$T$ circuits and channels for verification \cite{a18}.

Vilmart \cite{v22} gave an equational theory for the balanced sum-over-paths over Toffoli and Hadamard gates and showed its completeness via translation of the phase-free ZH-calculus \cite{ww19}. The phase-free ZH-calculus crucially restricts the ZH-calculus \cite{bk18} to balanced generators, allowing the direct translation of its equational theory into balanced sums --- while the full ZH-calculus admits an encoding in the balanced sum-over-paths \cite{lwk20,v21}, existing equational theories \cite{bk18,bkmww23} necessarily make use of unbalanced generators and the \emph{average} rule. It was further shown that this equational theory is complete for operators over Clifford and $R_k:=\mathrm{diag}(1, \zeta_{2^k})$ gates through an embedding into the Toffoli+Hadamard fragment. However, their re-writing system lacks the desirable properties of confluence, normal forms, and a primitive equational theory for Clifford+$R_k$. Moreover, the question of a complete equational theory for the sum-over-paths over $\C$ was left open, as well as the development of a direct analogue of the full ZH-calculus in the sum-over-paths model.

In this work we address these questions, introducing a concrete representation of the \emph{unbalanced} sums-over-paths and give complete equational theories for such sums over rings and fields. We consider general rings as fault-tolerantly constructible circuits are typically restricted to linear operators over subrings of $\mathbb{C}$ \cite{agr20}. Inspired by the ZH-calculus, completeness is attained by symbolically re-writing a sum to a unique normal form explicitly encoding the matrix entries. In essence, our system internalizes the method of falling back to explicit evaluation in cases where no progress can be made with re-writing \cite{a18}. As a result, the practical applicability of our equational theory is limited --- our goals are instead to provide representations and complete theories for general sums-over-paths so that effective re-writing systems may be further developed. 

The paper is organized as follows. In \cref{sec:background} we review the balanced sum-over-paths and its equational theories. In \cref{sec:unbalanced} we define a representation of the unbalanced sum-over-paths and give an equational theory which is complete for arbitrary rings. In \cref{sec:addition} we give a weaker equational theory which is not sound over rings, but is shown to be complete over any field.

\section{The balanced sum-over-paths}\label{sec:background}

In the path integral viewpoint, the action of a linear operator $\Psi:\H_1 \rightarrow \H_2$ between Hilbert spaces $\H_1$ and $\H_2$ on a state $\ket{i}$ of some orthonormal basis $\{\ket{i}\}$ of $\H_1$ can be described as a sum over some collection $\Pi$ of paths, where the path $\pi\in \Pi$ has amplitude $\psi(\pi)\in\C$ and ends in a state $\ket{f(\pi)}$ of an orthonormal basis of $\H_2$:
\[
	\Psi\ket{i} = \sum_{\pi\in \Pi} \psi(\pi)\ket{f(\pi)}.
\]
In contrast to the operator representation $\Psi\ket{i} = \sum_{j}\alpha_{ij}\ket{j}$, in general there may be many superimposed paths leading to a particular basis state, their amplitudes adding in these cases and resulting in interference. Likewise, the sequential composition of two operators is described by composing paths along their endpoints and multiplying the amplitude along each segment, in essence delaying evaluation of any interfering paths. This provides flexibility in the evaluation of individual amplitudes, but on the flip side requires effective means of representing and reasoning about a system of paths to be useful.

In \cite{a18} a concrete representation of \emph{amplitude-balanced} sums over $2^n$-dimensional Hilbert spaces was given via multilinear polynomials. An amplitude-balanced sum is one in which every path $\pi$ with non-zero amplitude has equal magnitude but may vary in the phase. By restricting to balanced sums, individual path amplitudes are described by unit-norm complex numbers, whose multiplicative group is isomorphic to the additive group $\mathbb{R}/2\pi$. In particular, $\psi(\pi) \approx e^{2\pi iP(\pi)}$ for some $P:\Pi \rightarrow \R/2\pi$ which has a unique representation as a multilinear polynomial, and if $\phi(\pi')  \approx e^{2\pi iQ(\pi')}$, then $\psi(\pi)\cdot \phi(\pi')  \approx e^{2\pi i [P(\pi) + Q(\pi')]}$ --- i.e. the phase along a composite path --- is uniquely representable in polynomial time. We review this representation below.

\begin{definition}[Balanced sum-over-paths]
A balanced path sum is an expression of the form
\[
	\Psi\ket{\vec{x}} = \mathcal{N}\sum_{\vec{y}\in\Z_2^{k}}e^{2\pi i P(\vec{x}, \vec{y})}\ket{f(\vec{x}, \vec{y})},
\]
where $\mathcal{N}\in\mathbb{C}$ and $P:\Z_2^m\times\Z_2^k\rightarrow \R/2\pi$ and $f:\Z_2^m\times\Z_2^k\rightarrow \Z_2^n$ are represented as a multilinear polynomial and a sequence of $n$ multilinear polynomials in $m+k$ variables, respectively.
\end{definition}

We use $\ket{\Psi}$ to denote the sum-over-paths representation of $\Psi$, or just $\Psi$ when it is clear from the context that we mean the symbolic expression rather than the linear operator. The variables appearing in a sum-over-paths expression $\ket{\Psi}$ which are not summed over are called \emph{free variables}. We denote the set of free variables of $\Psi$ by $FV(\Psi)$ and for the purpose of substitution use the notation $\ket{\Psi(x)}$ to identify a variable which may appear free in $\ket{\Psi}$. Specifically, given a sum-over-paths $\ket{\Psi(x)}$, $\ket{\Psi(f)}$ denotes the (capture-avoiding) substitution of $f$ for every free occurrence of $x$ in $\ket{\Psi}$, of which there may be none. We say a sum $\ket{\Psi}$ is \emph{closed} if $FV(\Psi) = \emptyset$, in which case $\ket{\Psi}$ corresponds to a vector. As a convention, we often use variable names $x_i$ to denote the free variables of a path sum and $y_i$, $z_i$ to denote variables which are summed, though it should be understood that this is not a rule and sum-over-paths expressions may contain arbitrary variables in free or summed positions. We further use letters $f$, $g$, $h$ to refer to Boolean-valued functions or symbolic expressions, and uppercase letters $P$, $Q$, $R$ to refer to symbolic expressions in other rings.

By linearity, compatible balanced sums may be sequentially composed through variable substitution:
\[
	\Phi\Psi\ket{\vec{x}} = \mathcal{N}\sum_{\vec{y}\in\Z_2^{k}}e^{2\pi i P(\vec{x}, \vec{y})}\Phi\ket{f(\vec{x}, \vec{y})}.
\]
As substitution involves substituting symbolic expressions over $\Z_2$ in the phase, a \emph{lifting} construction \cite{a18} is used to embed polynomial arithmetic over $\mathbb{Z}_2$ into polynomials over $\R$ (or more generally, any unital ring $\mathcal{R}$). In particular, we define the lifting $\overline{\cdot}$ of $(\mathbb{Z}_2, \oplus, \cdot)$ into $(\mathcal{R}, +, \cdot)$ recursively by
\begin{align*}
	\hspace*{5em} \overline{0} &= 0_\mathcal{R} \hspace{-5em} & \overline{f \cdot g} &= \overline{f} \cdot \overline{g} \\
	\hspace*{2em} \overline{1} &= 1_\mathcal{R}  \hspace{-5em} & \overline{f \oplus g} &= \overline{f} + \overline{g} - (2 \cdot \overline{f} \cdot \overline{g}).
\end{align*}
The tensor product, or parallel composition, can be defined via
\[
	(\Psi\otimes\Phi)(\ket{\vec{x}}\otimes \ket{\vec{y}}) = \Psi\ket{\vec{x}}\otimes\Phi\ket{\vec{y}},
\]
where the phases of $\Psi\ket{\vec{x}}$ and $\Phi\ket{\vec{y}}$ are multiplied and their final states are concatenated.
\begin{example}
  \label{ex:cliffordt2}
	The gates $\xgate$ and $\tgate$ admit the following representations as balanced sums:
 	 \begin{itemize}
  	\item $\xgate\ket{x} = \ket{1\oplus x}$ and
 	 \item $\tgate\ket{y} = \omega^y\ket{y}$, where $\omega=e^{\frac{2\pi i}{8}}$.   
 	 \end{itemize}
	Composing $\tgate$ after $\xgate$, we can compute the sum-over-paths expression of $\tgate\xgate$ by substituting $y$ in the expression of $\tgate$ with $x\oplus 1$, which lifts to $\overline{1\oplus x} = 1 - x$ in the exponent of $\omega$:
\[
	\tgate\xgate\ket{x} =  \omega^{1 - x}\ket{1 \oplus x}
\]
\end{example}

While the paths in a balanced path sum with non-zero amplitude all have the same magnitude $\mathcal{N}$, matrices which have entries of different magnitudes can be represented as balanced sums with interfering paths, as the following example illustrates.

\begin{example}
	The controlled-Hadamard gate
	\[
		\Lambda(\hgate) = \ket{0}\bra{0}\otimes I + \ket{1}\bra{1}\otimes H = \frac{1}{\sqrt{2}}\begin{bmatrix} \sqrt{2} & 0 & 0 & 0 \\ 0 & \sqrt{2} & 0 & 0 \\ 0 & 0 & 1 & 1 \\ 0 & 0 & 1 & -1 \end{bmatrix}
	\]
	when viewed as a collection of \emph{unique} transitions (i.e. non-interfering) in the computational basis is necessarily unbalanced. However, the $\Lambda(H)$ gate may be represented as a balanced sum-over-paths by using the control bit to cause the intermediate paths to interfere when it is in the $0$ state:
	\[
		\Lambda(\hgate)\ket{x_1x_2} = \frac{1}{\sqrt{2}} \sum_{y\in\Z_2} \omega^{(1 - x_1)(2y - 1)}(-1)^{x_1x_2y}\ket{x_1}\ket{(1 \oplus x_1)x_2 \oplus x_1y}
	\]
	Note that when $x_1 = 0$ we have $\frac{1}{\sqrt{2}} \sum_{y\in\Z_2} \omega^{(2y - 1)}\ket{0}\ket{x_2} = \frac{\omega + \omega^\dagger}{\sqrt{2}}\ket{0}\ket{x_2} = \ket{0}\ket{x_2}$, while when $x_1 = 1$ we have the sum $\frac{1}{\sqrt{2}} \sum_{y\in\Z_2} (-1)^{x_2y}\ket{1}\ket{y}$
	which is the representation of $\ket{1}\otimes (H\ket{x_2})$.

\end{example}

It has been shown that the balanced sum-over-paths is universal for linear operators over qubit ($2^n$-dimensional) Hilbert spaces, via a translation from the universal ZH-calculus \cite{v21}. Below we give a model of the universal ZX-calculus \cite{cd08} which is more natural to specify in the balanced sum-over-paths as the ZX-calculus is generated by balanced operators, while the $\hgate$-boxes of the ZH-calculus are unbalanced when the amplitude is non-unital.

\begin{example}
A simple model of the ZX-calculus via balanced sums can be defined over the universal generating set consisting of the $Z$-spider and Hadamard as so:
\begin{align*}
\left\llbracket
\begin{ZX}
\leftManyDots{n} \zxZ{\alpha} \rightManyDots{m} 
\end{ZX} \right\rrbracket\ket{\vec{x}} = \frac{1}{2^n}\sum_{\vec{y}\in\Z_2^n,z\in\Z_2}\alpha^z(-1)^{\sum_{i=1}^n y_i(x_i + z)}\ket{zz\cdots z}  \qquad 
\left\llbracket\begin{ZX}[ampersand replacement=\&]
\zxNone{} \rar \& \zxH{} \rar \& \zxNone{}
\end{ZX}\right\rrbracket\ket{x} = \frac{1}{\sqrt{2}}\sum_{y\in\Z_2}(-1)^{xy}\ket{y}
\end{align*}
Note that the sum for the $Z$-spider forces any path with non-zero amplitude to satisfy $z=x_1=x_2=\cdots=x_n$, since for any $i$, $\sum_{y_i\in\Z_2}(-1)^{y_i(x_i + z)} = 0$ whenever $z\neq x_i$, and $2$ otherwise. Encodings of an arbitrary $Z$-spider using fewer variables are possible, for example $\frac{1}{2^n}\sum_{y,z\in\Z_2}\alpha^z(-1)^{y(1 + \prod_{i=1}^n (x_i + z + 1))}\ket{zz\cdots z}
,$ but have size exponential in $n$.

\end{example}

We use $\equiv_{T}$ to denote the equivalence of two path sums up to a theory $T$ defined by a set of (sound) equations, together with the congruence
\[
	 \ket{\Psi} \equiv_{T} \ket{\Phi} \implies \mathcal{N}\sum_{\vec{y}\in\Z_2^k}e^{2\pi i P(\vec{x},\vec{y})}\ket{\Psi} \equiv_{T} \mathcal{N}\sum_{\vec{y}\in\Z_2^k}e^{2\pi i P(\vec{x},\vec{y})}\ket{\Phi}
\]
We use $\equiv_{e}$ to denote equivalence up to an individual equation $e$. We say an equational theory $T$ is \emph{complete} for a subset $\mathcal{C}$ of path sums if whenever $\ket{\Psi},\ket{\Phi}\in \mathcal{C}$,
\[
	\Psi = \Phi \implies \ket{\Psi} \equiv_{T} \ket{\Phi}.
\]

\begin{figure}
\begin{align}
	\sum_{y\in\Z_2}\ket{\Psi} &\equiv 2\ket{\Psi} \tag*{(E)}\label{eq:e} \\
	\sum_{x,y\in\Z_2}(-1)^{y(x + f)}\ket{\Psi(x)} &\equiv 2\ket{\Psi(f)} \tag*{(H)}\label{eq:i} \\
	\sum_{y\in\Z_2}i^y(-1)^{yf}\ket{\Psi} &\equiv \omega\sqrt{2}(-i)^{\overline{f}}\ket{\Psi} \tag*{($\omega$)}\label{eq:u}
\end{align}
\caption{A Clifford-complete system of equations for the balanced sum-over-paths, denoted $\equiv_{\mathrm{Cliff}}$. In all rules above $y\notin FV(\Psi)$ and $f$ is some Boolean expression such that $x,y\notin FV(f)$.}
\label{fig:cliffrewrite}
\end{figure}

\Cref{fig:cliffrewrite} gives the \ref{eq:e}, \ref{eq:i}, and \ref{eq:u} rules of the sum-over-paths which define the equational theory $\equiv_{\mathrm{Cliff}}$. It was previously shown that $\equiv_{\mathrm{Cliff}}$ is complete for Clifford path sums \cite{v21,abr22}.

\begin{example}
	The following equalities are derivable:
	\begin{align*}
		\sum_{y\in\Z_2}\ket{\Psi(y)} &\equiv_{\mathrm{Cliff}} \sum_{y\in\Z_2}\ket{\Psi(y \oplus f)} \textrm{ where $f$ is Boolean and $y\notin FV(f)$} \\
		\frac{1}{\sqrt{2}} \sum_{y\in\Z_2} \omega^{(2y - 1)}\ket{0}\ket{x_2} &\equiv_{\mathrm{Cliff}} \ket{0}\ket{x_2} \\
\left\llbracket (\begin{ZX}[ampersand replacement=\&]
\zxNone{} \rar \& \zxH{} \rar \& \zxNone{}
\end{ZX})^{\otimes m} \circ
\begin{ZX}
\leftManyDots{} \zxZ{\alpha} \rightManyDots{} 
\end{ZX} \circ (\begin{ZX}[ampersand replacement=\&]
\zxNone{} \rar \& \zxH{} \rar \& \zxNone{}
\end{ZX})^{\otimes n}\right\rrbracket\ket{\vec{x}}&\equiv_{\mathrm{Cliff}} \frac{1}{2^{n+m}}\sum_{y\in\Z_2,\vec{z}\in\Z_2^m}\alpha^y(-1)^{\sum_{i=1}^n x_iy}(-1)^{\sum_{j=1}^m yz_i}\ket{\vec{z}}
	\end{align*}
	The first equation is the \emph{variable change} rule from the Clifford-complete equational theory of \cite{v21}, whose derivation by the \ref{eq:e} rule was shown in \cite{abr22}. The second encodes the evaluation of $\Lambda(\hgate)\ket{0}\ket{x_2} = \ket{0}\ket{x_2}$ and follows from a single application of \ref{eq:u}. The third equation models the $X$-spider of the ZX-calculus by the color change law \cite{cd08}. 
\end{example}

\paragraph{Interference, algebraic varieties, and completeness for Toffoli+Hadamard}

As noted in \cite{a18}, \ref{eq:i} arises as an instance of a general (binary) \emph{interference} rule,
\begin{equation}
	\sum_{y\in\Z_2}(-1)^{yF}\ket{\Psi} \equiv 2\ket{\Psi\raisebox{-.5ex}{\rule{.4pt}{1.5ex}}{\:}_{F=0}}. \tag*{(I)}\label{eq:igen}
\end{equation}
where $F$ is a polynomial over $\Z_2$. Viewing $F$ as a proposition on the paths indexed by $FV(F)$, the sum $\sum_{y\in\Z_2}(-1)^{yF}\ket{\Psi}$ filters out paths satisfying $F$, while paths which do not satisfy $F$ pass through with double amplitude. However, to write the restriction $\ket{\Psi\raisebox{-.5ex}{\rule{.4pt}{1.5ex}}{\:}_{F=0}}$ as an expression, $\ket{\Psi} = \mathcal{N}\sum_{\vec{y}}e^{2\pi i P(\vec{x}, \vec{y})}\ket{f(\vec{x}, \vec{y})}$ must be expressed as a sum over the solutions of the equation $F(\vec{x},\vec{y}) = 0$. Recall that the algebraic variety $\mathcal{V}(I)$ of a polynomial ideal $I$ consists of all points $(a_1,\dots,a_k)$ such that $f(a_1,\dots, a_k) = 0$ for every polynomial $f$ in $I$. We may hence write the restricted sum as
\[
	\mathcal{N}\sum_{(\vec{x},\vec{y})\in\mathcal{V}(I)}e^{2\pi i P(\vec{x}, \vec{y})}\ket{f(\vec{x}, \vec{y})},
\]
where $I = \langle F \rangle$. Note that $P$ and $f$ may be canonically written modulo the ideal $I$ using Gr\"{o}bner bases, though the resulting re-write system is not an equational theory in the sense we consider here. Instead we may restrict $f$ to cases which can be solved by substitution. The simplest such cases are when $F = 0$ which is solved trivially for any point in the variety, and when $F = x + g$, $x\notin FV(g)$, which is solved by setting $x = g$. These two cases result in the \ref{eq:e} and \ref{eq:i} rules. Moreover, both equations are \emph{complete} relative to the variety $\mathcal{V}(I)$ in that they completely characterize the solutions to $f$.

In \cite{v22} Vilmart gave a complete equational theory for path sums over Toffoli and Hadamard, via restricted cases of \ref{eq:igen}. Such sums can be written in the form
\[
	\Psi\ket{\vec{x}} = \frac{1}{\sqrt{2^k}}\sum_{\vec{y}}(-1)^{P(\vec{x}, \vec{y})}\ket{f(\vec{x}, \vec{y})}.
\]
We define $\equiv_{\mathrm{TH}}$ to be equivalence of balanced sums up to $\equiv_{\mathrm{Cliff}}$ as well as the additional rules of \cref{fig:threwrite}. As noted in \cite{v22}, all three new equations arise as particular instances of the binary interference rule. The \ref{eq:ii} rule, which subsumes \ref{eq:i}, arises from \ref{eq:igen} when $F = x\cdot g + g\cdot f + 1$, in which case $x + g + 1 \in \langle F \rangle$ and hence $x = g + 1$ is a partial solution to $F = 0$. Likewise, \ref{eq:ii2} arises from the intersection of the varieties $f=0$ and $g=0$, where since $\langle f, g\rangle = \langle f + g + fg \rangle$ over $\Z_2$ the two equations to be combined into a single one without affecting the variety. The \ref{eq:z} rule arises when $F = 1$ and hence the variety is empty.
\begin{figure}
\begin{align}
	\sum_{y\in\Z_2}\sum_{x\in\Z_2}(-1)^{y(x\cdot g + g\cdot f + 1)}\ket{\Psi(x)} &\equiv \sum_{y\in\Z_2}(-1)^{y(g + 1)}\ket{\Psi(1 + f)} \tag*{(Hgen)}\label{eq:ii} \\
	\sum_{y\in\Z_2}\sum_{x\in\Z_2}(-1)^{y\cdot f +x\cdot g}\ket{\Psi} &\equiv 2\sum_{y\in\Z_2}(-1)^{y(f + g + f\cdot g)}\ket{\Psi} \tag*{(Hrel)}\label{eq:ii2}\\
	\sum_{y\in\Z_2}\sum_{x\in\Z_2}(-1)^{y}\ket{\Psi} &\equiv 0\ket{\Psi} \tag*{(Z)}\label{eq:z}
\end{align}
\caption{A complete equational theory $\equiv_{\mathrm{TH}}$ for path sums over Toffoli and Hadamard gates \cite{v22}. In all rules $y\notin FV(\Psi)$ and $f,g$ are Boolean expressions such that $x,y\notin FV(f)\cup FV(g)$}
\label{fig:threwrite}
\end{figure}

\section{The unbalanced sum-over-paths}\label{sec:unbalanced}

While computationally efficient for many problems, balanced sums are unwieldy for reasoning about vectors and matrices with entries of varying magnitude. Such states often arise in probabilistic quantum computations and algorithms, such as Grover's search or Shor's algorithm. Moreover, canonical forms for such operators are difficult to define and test for equality, as the following example illustrates.

\begin{example}
Consider the unit vector
\[
	\ket{\psi} = \frac{1}{\sqrt{1 + p^2}}(\ket{0} + p\ket{1})
\]
where $p$ is an odd prime. Any representation of $\ket{\psi}$ by a balanced sum with $\pm 1$ phases must satisfy $\braket{0|\psi} = 1$, $\braket{1|\psi} = p$, and hence must consist of at least $p+1$ distinct paths. Now let $f$ and $g$ be Boolean expressions in free variables $\{y_i\}$ with $1$ and $p$ satisfying assignments, respectively. Then
\[
	\frac{1}{2\sqrt{1 + p^2}}\sum_{x\in\Z_2}\sum_{\vec{y}\in\Z_2^{k}}\sum_{z\in\Z_2}(-1)^{z[(1-x)(1 + f(\vec{y})) + x(1 + g(\vec{y}))]}\ket{x}
\]
is a valid representation of $\ket{\psi}$, as is the representation above where $f$ and $g$ are replaced with any other Boolean expressions with the same number of solutions. By \ref{eq:igen}, we can re-write this sum over the variety generated by the ideal $I=\langle(1-x)(1 + f(\vec{y})) + x(1 + g(\vec{y}))\rangle$ as $\ket{\Psi} = \sum_{(x,\vec{y})\in \mathcal{V}(I)}\ket{x}$. However, if we take this as a normal form it is surely not unique, as any other variety with the same number of points for each $x$ --- for instance, any variety $\mathcal{W}$ equal up to a permutation of the coordinates of $\mathcal{V}$ --- gives the same operator.

\end{example}

In order to allow the natural representation of linear algebraic objects with unbalanced magnitudes, we now develop a generalization to amplitudes which may be expressed as Boolean powers of elements taken from some ring $\mathcal{R}$. Recall that integer powers may be defined in any unital ring $\mathcal{R}$ as
\[
	0^0 := 1, \qquad r^0 := 1, \qquad r^n := r\cdot r^{n-1}.
\]

Our language of sums is comprised of expressions of three types: Boolean expressions used to denote paths, $\mathcal{R}$-valued expressions, and linear operators over $\mathcal{R}$ in the computational basis.

\begin{definition}[Unbalanced sum-over-paths]\label{eq:unbalanced}
An unbalanced sum-over-paths is an expression $\ket{\Psi}$ of the following language
\begin{align*}
		f &::= 0 \mid 1 \mid x \mid f_1\cdot f_2 \mid f_1\oplus f_2 \mid \neg f := 1 \oplus f\\
		r &::= \alpha, \beta \in \mathcal{R} \mid r^f \mid r_1r_2 \mid r_1 + r_2 \\
		\ket{\Psi} &::= \sum_{\vec{y}} r \ket{f_1\cdots f_n}.
\end{align*}
Variables $x,y, z,\dots$ in all types of expressions range over Boolean values $\Z_2$ and $\mathcal{R}$ is a commutative ring.
\end{definition}
Expressions $f$, $r$, and $\ket{\Psi}$ are referred to as Boolean, $\mathcal{R}$, and sum-over-paths expressions, respectively. We denote by $FV(f)$, $FV(r)$, and $FV(\Psi)$ the free variables appearing in the given expression. As with balanced sums, variables $\vec{y}$ which are summed over are bound and hence not included in the set of free variables of an unbalanced sum. An expression is \emph{closed} if it contains no free variables. Note that closed unbalanced sums correspond to vectors.

$\mathcal{R}$-expressions in free variables $\{x_i\}$ may be interpreted as a non-standard syntax for the polynomial ring $\mathcal{R}[x_1,\dots, x_k]/\langle x_1^2 - x_1, \dots, x_k^2 - x_k\rangle$ which favours multiplication over addition in terms of computational efficiency. This allows our representation to coincide with balanced sums when possible, allowing a balanced representation to be used and manipulated normally, but providing an ``escape hatch'' in the form of ring sums. In \cref{sec:addition} we consider representations where the ring sum is dropped entirely.

In the balanced sum-over-paths it's generally not obvious how a given matrix $A\in\mathcal{M}_{n \times m}(\mathcal{R})$ may be represented directly, hence universality is achieved via the construction of a known universal set of generators. By contrast, unbalanced sums allow the direct representation of $A$, as the amplitude function $\psi(\vec{x},\vec{y}) = \bra{\vec{y}}A\ket{\vec{x}}$ can be written directly as an $\mathcal{R}$-expression. We first define the notation $\vec{x} = \vec{y}$ as shorthand for the bitwise equality of $\vec{x}$ and $\vec{y}$,
\[
\vec{x} = \vec{y} := \prod_i x_i \oplus \neg y_i.
\]
Then we may write $A$ as an unbalanced sum of the following form:
\[
	A\ket{\vec{x}} = \sum_{\vec{y}}\alpha_{00\cdots 0}^{00\cdots 0 = \vec{x}\vec{y}}\alpha_{00\cdots 1}^{00\cdots 1 = \vec{x}\vec{y}}\cdots\alpha_{11\cdots 1}^{11\cdots 1 = \vec{x}\vec{y}}\ket{\vec{y}}
\]
where $\vec{x}\vec{y}$ denotes the concatenation of $\vec{x}$ and $\vec{y}$, and $\alpha_{\vec{x}\vec{y}}$ is equal to $\bra{\vec{y}}A\ket{\vec{x}}$. Intuitively, for a given value of $\vec{y}$ and $\vec{x}$, the (exponential-size) $\mathcal{R}$-expression above evaluates to $\alpha_{\vec{x}\vec{y}} = \bra{\vec{y}}A\ket{\vec{x}}$, as the exponent of every of other $\alpha_{\vec{x}'\vec{y}'}$ evaluates to zero. We write out the product explicitly rather than as $\Pi_{\vec{z}}\alpha_{\vec{z}}^{\vec{z} = \vec{x}\vec{y}}$ so as to avoid confusion with the use of $\Pi$ as mathematical syntax and $\Sigma$ as a syntactical element of path sums.

\begin{proposition}[Universality]
	Any linear operator $\mathcal{R}^{2^n} \rightarrow \mathcal{R}^{2^m}$ can be expressed as an unbalanced path sum over $\mathcal{R}$.
\end{proposition}

\begin{example}
The $\Lambda(\hgate)$ gate can be expressed as the unbalanced sum
\[
	\Lambda(\hgate)\ket{x_1x_2} = \sum_{y} 0^{\neg x_1(x_2 \oplus y)}(1/\sqrt{2})^{x_1}(-1)^{x_1x_2y}\ket{x_1y}.
\]
\end{example}

\definecolor{colorZxZ}{RGB}{255,255,255}
\definecolor{colorZxX}{RGB}{255,255,255}
\definecolor{colorZxH}{RGB}{255,255,255}
\tikzset{
my box/.style={inner sep=1pt, draw, thick, fill=white,anchor=center},
}

\begin{example}
While the generalized $\hgate$-boxes of the ZH-calculus \cite{bk18} admit an indirect encoding in the balanced sum-over-paths via Euler angles \cite{v21},
the ZH-calculus can be directly encoded in the unbalanced sum-over-paths as below.
\begin{align*}
\left\llbracket
\begin{ZX}[/zx/user overlay nodes/.style={
zxZ/.append style={fill=none},zxX/.append style={rectangle}}
]
\leftManyDots{n} \zxZ{} \rightManyDots{m} 
\end{ZX} \right\rrbracket\ket{\vec{x}} &= \sum_{\vec{y}\in\Z_2^n}\sum_{z\in\Z_2}2^{-n}(-1)^{\sum_{i=1}^n y_i(x_i + z)}\ket{zz\cdots z}  \qquad
\left\llbracket\begin{ZX}[ampersand replacement=\&]
\leftManyDots{n} |[my box]| \alpha \rightManyDots{m} 
\end{ZX}\right\rrbracket\ket{\vec{x}} = \sum_{\vec{y}\in\Z_2^n}\alpha^{x_1\cdots x_ny_1\cdots y_m}\ket{\vec{y}}
\end{align*}
\end{example}

As is customary, to simplify the notation and proofs we view a linear operator $A:\mathcal{R}^n\rightarrow \mathcal{R}^m$ as a vector $A\in\mathcal{R}^{nm}$ via the channel-state duality and define normal forms only on closed sums. Note that $\eta = \sum_{y}\ket{yy}$ and its adjoint $\epsilon\ket{xy} = \frac{1}{2}\sum_{z}(-1)^{z(x \oplus y)}$, i.e. a ``cup and cap,'' are well-defined over any unital ring $\mathcal{R}$, hence we can move between the operator and vector view freely.

\begin{definition}[Normal form]
	A \emph{normal form} is a closed, unbalanced sum of the following form:
	\begin{equation} 
		\sum_{\vec{x}}\alpha_{00\cdots 0}^{00\cdots 0 = \vec{x}}\alpha_{00\cdots 1}^{00\cdots 1 = \vec{x}}\cdots\alpha_{11\cdots 1}^{11\cdots 1 = \vec{x}}\ket{\vec{x}}
	\end{equation}
\end{definition}

\begin{example}
The $\Lambda(\hgate)$ gate has the normal form below, with $0$'s suppressed:
\[
	\frac{1}{\sqrt{2}}\sum_{\vec{x}} \sqrt{2}^{0000=\vec{x}}\sqrt{2}^{0101=\vec{x}}1^{1010=\vec{x}}1^{1011=\vec{x}}1^{1110=\vec{x}}(-1)^{1111=\vec{x}}\ket{\vec{x}}.
\]
Note that $x_1$ and $x_2$ correspond to the first and second input bits, respectively, while $x_3$ and $x_4$ correspond to the first and second output.
\end{example}

We remark that normal forms are unique, which is a trivial consequence of the fact that they explicitly represent vectors in the computational basis by a sequence of $2^n$ amplitudes.
\begin{proposition}\label{prop:unique}
	Let $\Psi$ be a vector in $\mathcal{R}^{2^n}$. Then $\Psi$ has a unique normal form.
\end{proposition}

\subsection{Equational theory of unbalanced sums over $\mathcal{R}$}

\Cref{fig:rewrite} gives an equational theory, denoted $\equiv_{\mathcal{R}}$, for unbalanced sums. We separate equations into three classes: equations on Boolean expressions, $\mathcal{R}$-expressions, and equations involving sums. The equational theory of Boolean expressions is simply the well-known equational theory of commutative Boolean rings and is only provided for completeness. The equational theory of $\mathcal{R}$-expressions includes the axioms of commutative, unital rings on the left, and equations specific to $f$-powers on the right. The two rules involving path sums are the usual \ref{eq:iu} rule, and the new \emph{sum} rule \ref{eq:s} which internalizes sums over variables as sums of $\mathcal{R}$-expressions.

\begin{figure}[h]
\begin{subfigure}{\textwidth}
	\begin{minipage}[t]{0.2\textwidth}
	\begin{align*}
		f \oplus 0 &\equiv  f \\
		f \oplus f &\equiv  0
	\end{align*}
	\end{minipage}
	\begin{minipage}[t]{0.29\textwidth}
	\begin{align*}
		(f_1 \oplus f_2)  \oplus f_3 &\equiv  f_1  \oplus (f_2  \oplus f_3) \\
		f_1 \oplus f_2 &\equiv  f_2  \oplus f_1
	\end{align*}
	\end{minipage}
	\begin{minipage}[t]{0.2\textwidth}
	\begin{align*}
		f\cdot 1 &\equiv f \\
		f\cdot f &\equiv f
	\end{align*}
	\end{minipage}
	\begin{minipage}[t]{0.29\textwidth}
	\begin{align*}
		(f_1 \cdot f_2)  \cdot f_3 &\equiv  f_1  \cdot (f_2  \cdot f_3) \\
		f_1 \cdot f_2 &\equiv  f_2  \cdot f_1
	\end{align*}
	\end{minipage}
	\[
		f_1 \cdot (f_2 \oplus f_3) \equiv f_1 \cdot f_2 \oplus f_1\cdot f_3
	\]
\vspace{-1.5em}
	\caption{Rules for Boolean expressions}\label{fig:beqs}
\end{subfigure}
\vspace{-.5em}
\begin{subfigure}{\textwidth}
	\begin{minipage}{0.48\textwidth}
	\begin{align*}
		(r_1 + r_2) + r_3 &\equiv r_1 + (r_2 + r_3) \\
		r_1 + r_2 &\equiv r_2 + r_1 \\
		r_1 + 0 &\equiv r_1 \\
		r - r &\equiv 0 \\
		(r_1\cdot r_2)\cdot r_3 &\equiv r_1\cdot (r_2\cdot r_3) \\
		r_1 \cdot r_2 &\equiv r_2\cdot r_1 \\
		r \cdot 1 &\equiv r \\
		r_1\cdot (r_2 + r_3) &\equiv r_1\cdot r_2 + r_1\cdot r_3
	\end{align*}
	\end{minipage}
	\begin{minipage}{0.48\textwidth}
	\begin{align*}
		r^0 &\equiv 1 \equiv 1^f \\
		r^1 &\equiv r \equiv r^{f}r^{\neg f} \\
		r^{f_1\oplus f_2} &\equiv r^{f_1} + r^{f_2} - (2r)^{f_1\cdot f_2} \\
		r^{f_1\cdot f_2} &\equiv (r^{f_1})^{f_2} \\
		r_1^fr_2^f &\equiv (r_1r_2)^f \\
		r_1^fr_2^{\neg f} &\equiv r_1^f + r_2^{\neg f} - 1 \\
		r_1^f + r_2^f &\equiv (r_1+r_2)^f + 0^{\neg f}
	\end{align*}
	\end{minipage}
	\caption{Rules for $\mathcal{R}$-expressions.}\label{fig:reqs}
\end{subfigure}
\vspace{-.5em}
\begin{subfigure}{\textwidth}
\begin{align}
	\sum_{x,y}(-1)^{y(x \oplus f)} \ket{\Psi(x)} &\equiv 2 \ket{\Psi(f)}
		\tag*{(H)}\label{eq:iu} \\
	\sum_{y}r(y)\ket{\Psi} &\equiv (r(0) + r(1))\ket{\Psi}
		\tag*{(S)}\label{eq:s}
\end{align}
\vspace{-.5em}
\caption{Rules for sum-over-paths expressions. In-scope variables are not free in any sub-expressions unless explicitly included in parentheses.}\label{fig:seqs}
\end{subfigure}
\caption{Equational theory $\equiv_{\mathcal{R}}$ for unbalanced sums over rings $\mathcal{R}$.}
\label{fig:rewrite}
\vspace{-.5em}
\end{figure}

To show completeness, we proceed by first defining a normal form for $\mathcal{R}$-expressions and showing that every $\mathcal{R}$-expression can be re-written in normal form. This forms the bulk of the proof, as the \ref{eq:s} rule can be used to force the evaluation of any internal variable by the $\mathcal{R}$-expression sub-language.

\begin{definition}[$\mathcal{R}$-expression normal form]\label{def:nf}
An $\mathcal{R}$-expression over the variables $\{x_i\}$ is in \emph{normal form} if it is of the form
\[
	\alpha_{00\cdots 0}^{00\cdots 0 = \vec{x}}\alpha_{00\cdots 1}^{00\cdots 1  = \vec{x}}\cdots \alpha_{11\cdots 1}^{ 11\cdots 1  = \vec{x}}.
\]
\end{definition}

\begin{proposition}[$\mathcal{R}$-expression normalization]\label{prop:rexp}
An $\mathcal{R}$-expression $r$ can be brought into normal form over the variables $\{x_i\}\supseteq FV(r)$ using the equations of \cref{fig:rewrite}.
\end{proposition}

A proof of \cref{prop:rexp} is given in \cref{app:proof}.
We next turn our attention to normalization of expressions involving sums. Normalization proceeds by writing the closed sum as a sum over all basis vectors by equating outputs with fresh variables, then summing along each internal variable and normalizing the resulting $\mathcal{R}$-expression.

\begin{theorem}\label{thm:complete}
	$\equiv_{\mathcal{R}}$ is complete for unbalanced sums over any commutative ring $\mathcal{R}$.
\end{theorem}
\begin{proof}
	Let $\ket{\Psi} = \sum_{\vec{x}}r\ket{f_1f_2\cdots f_n}$ be a closed, unbalanced sum. Then
		\begin{align*}
			\sum_{\vec{x}}r\ket{f_1f_2\cdots f_n} 
				&\equiv \sum_{\vec{x}}\sum_{\vec{y}}\sum_{\vec{z}}r(-1)^{y_1(z_1 \oplus f_1) + y_2(z_2 \oplus f_2) + \cdots + y_n(z_n \oplus f_n)}\ket{z_1 z_2\cdots z_n} & \text{by \ref{eq:i}} \\
				&\equiv \sum_{\vec{z}}\sum_{\vec{x}}\sum_{\vec{y}}r'(\vec{x},\vec{y})\ket{\vec{z}} \\
				&\equiv \sum_{\vec{z}}(r'(00\cdots 0) + r'(00\cdots 1) + \cdots + r'(11\cdots 1))\ket{\vec{z}} & \text{by \ref{eq:s}} \\
				&\equiv \sum_{\vec{z}}\alpha_{00\cdots 0}^{00\cdots 0 = \vec{z}}\alpha_{00\cdots 1}^{00\cdots 1 = \vec{z}}\cdots\alpha_{11\cdots 1}^{11\cdots 1 = \vec{z}}\ket{\vec{z}} & \text{by \cref{prop:rexp}}
		\end{align*}
\end{proof}

\section{Weakening the sum rule}\label{sec:addition}

The equational theory developed in the preceding section is too strong for use in practice. Indeed, re-writing effectively amounts to explicit evaluation of the sum, e.g.,
	\[
		\sum_{x}\ket{\Psi(x)} = \ket{\Psi(0)} + \ket{\Psi(1)},
	\]
together with a set of rules for manipulating certain symbolic expressions over rings. This is made possible by the highly-expressive sub-language of $\mathcal
{R}$-expressions, which allows for the summation of arbitrary $\mathcal{R}$-expressions and hence the super-powered sum rule. However, sums of $\mathcal{R}$ expressions are difficult to re-write and generally require complete expansion of the expression to a normal form. In particular, with unrestricted use of the sum rule we are not likely to discover efficient proofs of equality. 

To limit the power of the sum rule, in this section we define a fragment of the unbalanced sum-over-paths which eliminates sums of \emph{symbolic} amplitudes, and give a complete equational theory over arbitrary fields.

\begin{definition}[Multiplicative sum-over-paths]\label{eq:unbalanced}
The multiplicative fragment of the unbalanced sum-over-paths over a field $\mathcal{F}$ consists of unbalanced sums of the form
\begin{align*}
		a &::= \alpha, \beta \in \mathcal{F} \mid a^f \mid a_1a_2 \\
		\ket{\Psi} &::= \sum_{\vec{y}} a \ket{f_1\cdots f_n}.
\end{align*}
Boolean expressions $f$ are defined as in the unbalanced sum-over-paths.
\end{definition}

We call amplitude expressions of the form $a$ $\mathcal{F}$-expressions. It can be readily observed that as normal forms live in the multiplicative fragment, the multiplicative fragment is again universal, and is equivalent to the full unbalanced sum-over-paths up to $\equiv_{\mathcal{R}}$.

\begin{figure}
\begin{subfigure}[b]{\textwidth}
	\begin{minipage}{0.33\textwidth}
	\begin{align*}
		(a_1\cdot a_2)\cdot a_3 &\equiv a_1\cdot (a_2\cdot a_3) \\
		a_1 \cdot a_2 &\equiv a_2\cdot a_1 \\
		a \cdot 1 &\equiv a
	\end{align*}
	\end{minipage}
	\begin{minipage}{0.33\textwidth}
	\begin{align*}
		a^0 &\equiv 1 \equiv 1^f \\
		a^1 &\equiv a \equiv a^{f}a^{\neg f}
	\end{align*}
	\end{minipage}
	\begin{minipage}{0.33\textwidth}
	\begin{align*}
		a^{f_1\oplus f_2} &\equiv a^{f_1}a^{f_2}(a^{-2})^{f_1\cdot f_2} \\
		a^{f_1\cdot f_2} &\equiv (a^{f_1})^{f_2} \\
		a_1^fa_2^f &\equiv (a_1a_2)^f
	\end{align*}
	\end{minipage}
	\caption{Rules for $\mathcal{F}$-expressions.}\label{fig:feqs}
\end{subfigure}
\begin{subfigure}[b]{\textwidth}
\begin{align}
	\sum_{x,y}(-1)^{y(x \oplus f)} \ket{\Psi(x)} &\equiv 2 \ket{\Psi(f)}
		\tag*{(H)}\label{eq:iii} \\
	\sum_{y,z}a_1^x(y)a_2^{\neg x}(z)\ket{\Psi(x)}
		&\equiv 2 \sum_{y}a_1^x(y)a_2^{\neg x}(y)\ket{\Psi(x)}
		\tag*{(O)}\label{eq:ooo} \\
	\sum_{y}(\alpha^y\beta^{\neg y})^f\ket{\Psi} &\equiv 2\left(\frac{\alpha + \beta}{2}\right)^f \ket{\Psi} \tag*{(A)}\label{eq:aaa}
\end{align}
\caption{Rules for unbalanced sums in the multiplicative fragment. In-scope variables are not free in sub-expressions unless explicitly listed in parentheses.}
\end{subfigure}
\caption{Equational theory $\equiv_{\mathcal{F}}$ of multiplicative sums over fields $\mathcal{F}$.}\label{fig:rewrite2}
\end{figure}

\Cref{fig:rewrite2} defines an equational theory, denoted $\equiv_{\mathcal{F}}$, for the multiplicative fragment which is defined and sound when $\mathcal{F}$ is a field. Note that the equation $a^{f_1\oplus f_2} \equiv a^{f_1}a^{f_2}(a^{-2})^{f_1\cdot f_2}$, which coincides with the lifting of $f_1\oplus f_2$ to $f_1 + f_2 - 2f_1\cdot f_2$, is not well defined in arbitrary rings. Otherwise, the equational theory coincides with the equational theory of $\mathcal{R}$-expressions with rules for sums of expressions removed. We omit the equational rules for Boolean expressions as they are the same as those of \cref{fig:beqs}.

\begin{proposition}[$\mathcal{F}$-expression normalization]\label{prop:fexp}
An $\mathcal{F}$-expression $a$ can be brought into normal form over the variables $\{x_i\}\supseteq FV(a)$ using the equations of \cref{fig:rewrite2}.
\end{proposition}
\begin{proof}
The proof follows identically to the proof of \cref{prop:rexp}. The one different case of $a^{f_1\oplus f_2}$ is handled similar to the corresponding case in \cref{prop:rexp} by the fact that $a^{f_1\oplus f_2} \equiv a^{f_1}a^{f_2}(a^{-2})^{f_1\cdot f_2}$ where each of $a^{f_1}$, $a^{f_2}$, and $(a^{-2})^{f_1\cdot f_2}$ can be normalized by the inductive hypothesis and the case of $a^{f_1\cdot f_2}$. The case is then finished by the normalization of products.
\end{proof}

The \ref{eq:aaa} rule, which is a transliteration of the average rule of the ZH-calculus, replaces the \ref{eq:s} rule of $\equiv_{\mathcal{R}}$. Rather than summing ``top-down'' as in the normalization of the unbalanced sum-over-paths, the average rule allows amplitudes to only be summed ``bottom-up,'' i.e. by summing pure elements of $\mathcal{F}$. In particular, we may think about the normalization of a sum in (almost) normal form over one internal variable:
\[
	\sum_{\vec{x}}\sum_{y}\alpha_{00\cdots 00}^{00\cdots 00 = \vec{x}y}\alpha_{00\cdots 01}^{00\cdots 01 = \vec{x}y}\cdots\alpha_{11\cdots 11}^{11\cdots 11 = \vec{x}y}\ket{\vec{x}}.
\]
Using the rules of $\equiv_{\mathcal{R}}$, we may factor out the $y$ exponents and evaluate the sum over $y$ top-down as
\begin{align*}
	\sum_{\vec{x}}\sum_{y}\alpha_{00\cdots 00}^{00\cdots 00 = \vec{x}y}\alpha_{00\cdots 01}^{00\cdots 01 = \vec{x}y}\cdots\alpha_{11\cdots 11}^{11\cdots 11 = \vec{x}y}\ket{\vec{x}}
		&\equiv_{\mathcal{R}} \sum_{\vec{x}}\sum_y(\alpha_{00\cdots 00}^{00\cdots 0 = \vec{x}}\cdots\alpha_{11\cdots 10}^{11\cdots 1 = \vec{x}})^{\neg y}(\alpha_{00\cdots 01}^{ 00\cdots 0 = \vec{x}}\cdots\alpha_{11\cdots 11}^{11\cdots 1 = \vec{x}})^{y}\ket{\vec{x}} \\
		&\equiv_{\mathcal{R}} \sum_{\vec{x}}(\alpha_{00\cdots 00}^{00\cdots 0 = \vec{x}}\cdots\alpha_{11\cdots 10}^{11\cdots 1= \vec{x}} + \alpha_{00\cdots 01}^{ 00\cdots 0= \vec{x}}\cdots\alpha_{11\cdots 11}^{11\cdots 1= \vec{x}})\ket{\vec{x}}
\end{align*}
Under the more restrictive rules of $\equiv_{\mathcal{F}}$, the $y$ exponents must be brought inwards and the amplitudes summed in pairs:
\begin{align*}
	\sum_{\vec{x}}\sum_{y}\alpha_{00\cdots 00}^{00\cdots 00 = \vec{x}y}\alpha_{00\cdots 01}^{00\cdots 01 = \vec{x}y}\cdots\alpha_{11\cdots 11}^{11\cdots 11 = \vec{x}y}\ket{\vec{x}}
		&\equiv_{\mathcal{F}} \sum_{\vec{x}}\sum_y(\alpha_{00\cdots 00}^{\neg y}\alpha_{00\cdots 01}^{y})^{00\cdots 0 = \vec{x}}\cdots(\alpha_{11\cdots 10}^{\neg y}\alpha_{11\cdots 11}^{y})^{11\cdots 1 = \vec{x}}\ket{\vec{x}}
\end{align*}
In order to be able to apply the \ref{eq:aaa} rule to eliminate $y$ with this factorization, a distinct variable is needed for each pair, which can be achieved with the \ref{eq:ooo} rule --- a transliteration of the ortho rule of the ZH-calculus --- since every pair's exponent varies in the polarity of at least one variable:
\begin{align*}
	&\sum_{\vec{x}}\sum_y(\alpha_{00\cdots 00}^{\neg y}\alpha_{00\cdots 01}^{y})^{00\cdots 0 = \vec{x}}\cdots(\alpha_{11\cdots 10}^{\neg y}\alpha_{11\cdots 11}^{y})^{11\cdots 1 = \vec{x}}\ket{\vec{x}} \\ \equiv_{\mathcal{F}}&\sum_{\vec{x}}\sum_{\vec{y}}\frac{1}{2^{2^n - 1}}(\alpha_{00\cdots 00}^{\neg y_1}\alpha_{00\cdots 01}^{y_1})^{00\cdots 0 = \vec{x}}\cdots(\alpha_{11\cdots 10}^{\neg y_n}\alpha_{11\cdots 11}^{y_n})^{11\cdots 1 = \vec{x}}\ket{\vec{x}} & \text{by \ref{eq:ooo}} \\
\equiv_{\mathcal{F}}&\sum_{\vec{x}}\frac{2^{2^n}}{2^{2^n - 1}}\left(\frac{\alpha_{00\cdots 00} + \alpha_{00\cdots 01}}{2}\right)^{00\cdots 0 = \vec{x}}\cdots\left(\frac{\alpha_{11\cdots 10} + \alpha_{11\cdots 11}}{2}\right)^{11\cdots 1 = \vec{x}}\ket{\vec{x}} & \text{by \ref{eq:aaa}}\\ \equiv_{\mathcal{F}}&\sum_{\vec{x}}(\alpha_{00\cdots 00}+\alpha_{00\cdots 01})^{00\cdots 0 = \vec{x}}\cdots(\alpha_{11\cdots 10}+\alpha_{11\cdots 11})^{11\cdots 1 = \vec{x}}\ket{\vec{x}}
\end{align*}
Note that in the final line above, each sum involves concrete values taken from $\mathcal{F}$ and hence can be evaluated explicitly over $\mathcal{F}$.

\begin{theorem}\label{thm:complete2}
	$\equiv_{\mathcal{F}}$ is complete for multiplicative sums over any field $\mathcal{F}$.
\end{theorem}
\begin{proof}
	Let $\ket{\Psi} = \sum_{\vec{x}}a\ket{f_1f_2\cdots f_n}$ be a closed, multiplicative sum and note that $\ket{\Psi} \equiv_{\mathcal{F}} \sum_{\vec{x}}\sum_{\vec{y}}a'\ket{\vec{x}}$ by \ref{eq:iii}. Then by \cref{prop:fexp}, $a'$ can be normalized to give
\[
	\ket{\Psi} \equiv_{\mathcal{F}} \sum_{\vec{x}}\sum_{\vec{y}}\alpha_{00\cdots 0}^{00\cdots 0 = \vec{x}\vec{y}}\alpha_{00\cdots 1}^{00\cdots 1 = \vec{x}\vec{y}}\cdots\alpha_{11\cdots 1}^{11\cdots 1 = \vec{x}\vec{y}}\ket{\vec{x}}.
\]
If $\vec{y}$ is empty, then we're done. Otherwise, we can remove one variable at a time with the \ref{eq:ooo} and \ref{eq:aaa} rules as above until no internal variables are left.
\end{proof}

\section{Discussion}

We have now given a concrete syntax for sum-over-paths expressions with unbalanced amplitudes. We gave equational theories for rings and fields, and showed that each is complete. While the equational theories we give here are simplistic and inefficient, our hope is that a complete equational theory will allow the development of effective, complete re-writing systems. 

The equational theories we have developed --- particularly $\equiv_{\mathcal{F}}$ --- as well as our normal forms can be viewed as translations of the ZH-calculus with varying levels of freedom in the expression of \emph{Boolean} data. In the ZH-calculus, propagation of Boolean expressions along wires is a \emph{semantic} property, while in the sum-over-paths it is \emph{syntactic}. This allows the \ref{eq:iii} rule to take a more general form allowing the substitution of a variable with an expression, whereas in the ZH-calculus this logic is spread across several different rules. 

As an exercise we could attempt to formulate a more ZH-like sum-over-paths by restricting the Boolean expression language further, e.g.,
\begin{align*}
		f &::= 0 \mid 1 \mid x,y,z,\dots \mid f_1\cdot f_2 \\
		a &::= \alpha \in \mathcal{R} \mid \alpha^{f} \mid a_1a_2 \\
		\ket{\Psi} &::= \sum_{v\in V} a \ket{x_1\cdots x_n},
\end{align*}
and reformulate our equational theory for such a language. One natural formulation of the \ref{eq:iii} rule is $\sum_{x,y}(-1)^{xy}(-1)^{yf}a(x) \ket{\Psi} \equiv 2 a(f) \ket{\Psi}$, which corresponds to a slightly more general version of the HS1 rule of the ZH-calculus, which would be more accurately translated as $\sum_{x,y}(-1)^{xy}(-1)^{yf}\alpha^{xg} \ket{\Psi} \equiv 2 \alpha^{fg} \ket{\Psi}$. With either formulation, a secondary rule is needed to cover propagation of Boolean sums, corresponding to expressions of the form
\[
	\sum_{x,y}(-1)^{xy}(-1)^{yf_1}(-1)^{yf_2}\cdots (-1)^{yf_k}a(x) \ket{\Psi}.
\]
Semantically, this is equivalent to the expression $2 a(f_1\oplus f_2 \oplus \cdots \oplus f_k) \ket{\Psi}$, but since the language cannot express Boolean sums, we must distribute $a$ over the sum in a single step. One natural option to avoid imposing restrictions on the underlying ring, as we did in this section, is to restrict $a$ to the phase-free fragment. In particular,
\[
	\sum_{x,y}(-1)^{xy}(-1)^{yf_1}\cdots (-1)^{yf_k}(-1)^{xg}\ket{\Psi} \equiv 2(-1)^{f_1g}\cdots(-1)^{f_kg}\ket{\Psi},
\] which is equivalent up to \ref{eq:iii} to the BA2 rule of the ZH-calculus.

On the one hand, it is unclear what the utility of such an exercise might be, beyond as a symbolic syntax for ZH diagrams. On the other hand, these investigations shed light on both the similarities and differences between the graphical and symbolic approach. Notably, the symbolic approach naturally allows highly expressive languages, and the use of existing theories developed within the framework of non-categorical algebra and symbolic computation. On the other hand, the use of an expressive symbolic language naturally makes it more challenging to apply \emph{local} reasoning than in graphical theories. For these reasons we hypothesize that the sum-over-paths approach may be more amenable to \emph{automated} reasoning, while the graphical approach may be more amenable to \emph{interactive} reasoning. We leave it for future work to explore this premise further.

\section{Acknowledgements}
The author wishes to thank Louis Lemonnier and Aleks Kissinger for bringing the connection to the ZH-calculus and a path sum formulation of the ortho rule to their attention. The author also wishes to thank Julien Ross for many helpful discussions, as well as the anonymous reviewers for comments which have improved the presentation of the paper. This work was supported by Canada's AARMS and NSERC.

\bibliographystyle{eptcs} 
\bibliography{completeSOP}

\appendix

\section{Proof of \cref{prop:rexp}}\label{app:proof}
\begin{proposition}[$\mathcal{R}$-expression normalization]
An $\mathcal{R}$-expression $r$ can be brought into normal form over the variables $\{x_i\}\supseteq FV(r)$ using the equations of \cref{fig:rewrite}.
\end{proposition}
\begin{proof}
By induction on the structure of an $\mathcal{R}$-expression. 

\paragraph{Case: $\alpha$.} By induction on the number of variables $m$. If $m=0$ then $\alpha$ is already in normal form. For $m > 0$, let $\alpha \equiv \alpha_{00\cdots 0}^{\vec{x} = 00\cdots 0}\cdots \alpha_{11\cdots 1}^{\vec{x} = 11\cdots 1}$ and observe that $r$ can be brought into normal form involving one additional variable $x_{m+1}$:
		\begin{align*}
			\alpha &\equiv \alpha_{00\cdots 0}^{\vec{x} = 00\cdots 0}\cdots \alpha_{11\cdots 1}^{\vec{x} = 11\cdots 1} \\
				&\equiv (\alpha_{00\cdots 0}^{\vec{x} = 00\cdots 0}\cdots \alpha_{11\cdots 1}^{\vec{x} = 11\cdots 1})^{\neg x_{m+1}}(\alpha_{00\cdots 0}^{\vec{x} = 00\cdots 0}\cdots \alpha_{11\cdots 1}^{\vec{x} = 11\cdots 1})^{x_{m+1}} \\
				&\equiv (\alpha_{00\cdots 0}^{\vec{x} = 00\cdots 0})^{\neg x_{m+1}}\cdots (\alpha_{11\cdots 1}^{\vec{x} = 11\cdots 1})^{\neg x_{m+1}}(\alpha_{00\cdots 0}^{\vec{x} = 00\cdots 0})^{x_{m+1}}\cdots(\alpha_{11\cdots 1}^{\vec{x} = 11\cdots 1})^{x_{m+1}} \\
				&\equiv \alpha_{00\cdots 0}^{(\vec{x} = 00\cdots 0)(x_{m+1} = 0)}\cdots \alpha_{11\cdots 1}^{(\vec{x} = 11\cdots 1)(x_{m+1} = 0)}\alpha_{00\cdots 0}^{(\vec{x} = 00\cdots 0)(x_{m+1} = 1)}\cdots\alpha_{11\cdots 1}^{(\vec{x} = 11\cdots 1)(x_{m+1} = 1)} \\
				&\equiv \alpha_{0\cdots 00}^{\vec{x}x_{m+1} = 00\cdots 00}\cdots \alpha_{11\cdots 1}^{\vec{x}x_{m+1} = 11\cdots 10}\alpha_{00\cdots 0}^{\vec{x}x_{m+1} = 00\cdots 01}\cdots\alpha_{11\cdots 1}^{\vec{x}x_{m+1} = 11\cdots 11}
		\end{align*}

\paragraph{Case: $r_1r_2$.} Let $r_1 \equiv  \alpha_{00\cdots 0}^{\vec{x} = 00\cdots 0}\alpha_{00\cdots 1}^{\vec{x} = 00\cdots 1}\cdots \alpha_{11\cdots 1}^{\vec{x} = 11\cdots 1}$ and $r_2 \equiv  \beta_{00\cdots 0}^{\vec{x} = 00\cdots 0}\beta_{00\cdots 1}^{\vec{x} = 00\cdots 1}\cdots \beta_{11\cdots 1}^{\vec{x} = 11\cdots 1}$. Then
		\begin{align*}
		r_1r_2 &\equiv  \alpha_{00\cdots 0}^{\vec{x} = 00\cdots 0}\cdots \alpha_{11\cdots 1}^{\vec{x} = 11\cdots 1}\beta_{00\cdots 0}^{\vec{x} = 00\cdots 0}\cdots \beta_{11\cdots 1}^{\vec{x} = 11\cdots 1} \\
			  &\equiv  \alpha_{00\cdots 0}^{\vec{x} = 00\cdots 0}\beta_{00\cdots 0}^{\vec{x} = 00\cdots 0}\cdots \alpha_{11\cdots 1}^{\vec{x} = 11\cdots 1} \beta_{11\cdots 1}^{\vec{x} = 11\cdots 1} \\
			  &\equiv  (\alpha_{00\cdots 0}\beta_{00\cdots 0})^{\vec{x} = 00\cdots 0} \cdots (\alpha_{11\cdots 1}\beta_{11\cdots 1})^{\vec{x} = 11\cdots 1}
		\end{align*}

\paragraph{Case: $r_1 + r_2$.} Let $r_1$ be written in normal form as $r_1 \equiv \alpha_{00\cdots 0}^{\vec{x} = 00\cdots 0}\alpha_{00\cdots 1}^{\vec{x} = 00\cdots 1}\cdots \alpha_{11\cdots 1}^{\vec{x} = 11\cdots 1}$ and factorize this as $s_1^{\neg x_m}s_2^{x_m}$. Likewise, let $r_2 \equiv  \beta_{00\cdots 0}^{\vec{x} = 00\cdots 0}\beta_{00\cdots 1}^{\vec{x} = 00\cdots 1}\cdots \beta_{11\cdots 1}^{\vec{x} = 11\cdots 1} \equiv  t_1^{\neg x_m}t_2^{x_m}$. Then we can observe:
		\begin{align*}
		r_1 + r_2 &\equiv  s_1^{\neg x_m}s_2^{x_m} + t_1^{\neg x_m}t_2^{x_m} \\
			  &\equiv  s_1^{\neg x_m} + s_2^{x_m} - 1 + t_1^{\neg x_m} + t_2^{x_m} - 1 \\
			  &\equiv  s_1^{\neg x_m} + s_2^{x_m} - 1 + t_1^{\neg x_m} + t_2^{x_m} - 1 \\
			  &\equiv  (s_1 + t_1)^{\neg x_m} + 0^{x_m} + (s_2 + t_2)^{x_m} + 0^{\neg x_m} - 2 \\
			  &\equiv  (s_1 + t_1)^{\neg x_m} + (s_2 + t_2)^{x_m} - 1 + 0^{\neg x_m} + 0^{x_m} - 1 \\
			  &\equiv  (s_1 + t_1)^{\neg x_m}(s_2 + t_2)^{x_m} + 0^{\neg x_m}0^{x_m} \\
			  &\equiv  (s_1 + t_1)^{\neg x_m}(s_2 + t_2)^{x_m} \\
		\end{align*}

Now $s_1 + t_1$ and $s_2 + t_2$ are $\mathcal{R}$-expressions in $m-1$ variables, so induction on the number of variables suffices to finish this case. Note that in the base case $m = 0$, $r_1 + r_2$ is an ordinary ring sum and hence can be evaluated in $\mathcal{R}$ to the normal form $\alpha$. 

\paragraph{Case: $r^{f}$.} We proceed by induction on the structure of $f$. If $f=0$ then $r^f\equiv 1$ which can be brought into normal form by the previous case. If $f=1$ then $r^f \equiv r$ which is already in normal form. If $f = x_i$ then
		\begin{align*}
			r^{x_i} &\equiv (\alpha_{00\cdots 0}^{\vec{x} = 00\cdots 0}\alpha_{00\cdots 1}^{\vec{x} = 00\cdots 1}\cdots \alpha_{11\cdots 1}^{\vec{x} = 11\cdots 1})^{x_i} \\
				&\equiv (\alpha_{00\cdots 0}^{\vec{x} = 00\cdots 0})^{x_i}(\alpha_{00\cdots 1}^{\vec{x} = 00\cdots 1})^{x_i}\cdots (\alpha_{11\cdots 1}^{\vec{x} = 11\cdots 1})^{x_i} \\
				&\equiv 1^{\vec{x} = 00\cdots 0} \cdots 1^{\vec{x} = 01\cdots 1} \alpha_{10\cdots 0}^{\vec{x} = 10\cdots 0} \cdots \alpha_{11\cdots 1}^{\vec{x} = 11\cdots 1}
		\end{align*}
		
		For the inductive cases, if $f = f_1\cdot f_2$, then
		\begin{align*}
			r^{f_1\cdot f_2} &\equiv (r^{f_1})^{f_2} \\
				&\equiv (\alpha_{00\cdots 0}^{\vec{x} = 00\cdots 0}\alpha_{00\cdots 1}^{\vec{x} = 00\cdots 1}\cdots \alpha_{11\cdots 1}^{\vec{x} = 11\cdots 1})^{f_2} \\
				&\equiv (\alpha_{00\cdots 0}^{\vec{x} = 00\cdots 0})^{f_2}(\alpha_{00\cdots 1}^{\vec{x} = 00\cdots 1})^{f_2}\cdots (\alpha_{11\cdots 1}^{\vec{x} = 11\cdots 1})^{f_2} \\
				&\equiv r_{00\cdots 0}r_{00\cdots 1}\cdots r_{11\cdots 1}
		\end{align*}
		where each $r_{\vec{x}}$ is in normal form and the $r_1r_2$ case suffices to finish this case.
		
		Finally, if $f \equiv f_1\oplus f_2$, then $r^{f_1\oplus f_2} \equiv r^{f_1} + r^{f_2} - (2r)^{f_1\cdot f_2}$ where each term can be brought into normal form by the inductive hypothesis and the previous case. The $r_1 + r_2$ case then completes this case.
\end{proof}

\end{document}